\newif\ifcompj\compjfalse\ifcompj
\documentclass{comjnl}
\else
\documentclass[11pt]{article}
\usepackage{fullpage}
\usepackage{times}
\fi

\usepackage[utf8]{inputenc}
\usepackage{graphicx}
\usepackage{booktabs}
\usepackage{array} 
\usepackage{multirow}
\usepackage{xfrac}

\usepackage{siunitx}

\usepackage{amsthm}
\usepackage{amsmath}
\usepackage{amsfonts}
\usepackage{amssymb}
\usepackage{mathtools}

\usepackage{thmtools}
\usepackage{thm-restate}

\usepackage{algpseudocode}
\usepackage{algorithm}

\usepackage[tableposition=top,font=footnotesize,labelfont=bf,labelsep=period]{caption}
\usepackage{subcaption}
\restylefloat{table}

\usepackage[hyphens]{url}
\usepackage[draft]{hyperref}

\usepackage{cleveref}

\newcommand{\Z}{\mathbb{Z}}
\newcommand{\F}{\mathbb{F}}
\newcommand{\hashfun}[1]{\texttt{#1}}
\newcommand{\muhash}{\hashfun{MuHash}}
\newcommand{\lthash}{\hashfun{LtHash}}
\newcommand{\adhash}{\hashfun{AdHash}}
\newcommand{\msetvaddhash}{\hashfun{MSet-VAdd-Hash}}
\newcommand{\msetaddhash}{\hashfun{MSet-Add-Hash}}
\newcommand{\msetmuhash}{\hashfun{MSet-Mu-Hash}}
\newcommand{\Expect}{{\rm I\kern-.3em E}}
\newcommand{\card}[1]{\lvert #1 \rvert}

\newcommand{\gf}[1]{\F_{2^{#1}}}
\newcommand{\cpuinstruction}[1]{\texttt{#1}}
\newcommand{\tablememory}[1]{s_{#1}}

\providecommand\given{} 
\newcommand\SetSymbol[1][]{
   \nonscript\,#1\vert\nonscript\,\mathopen{}\allowbreak}
\DeclarePairedDelimiterX\Set[1]{\lbrace}{\rbrace}%
 { \renewcommand\given{\SetSymbol[\delimsize]} #1 }

\renewcommand{\emptyset}{\varnothing}

\DeclareMathOperator{\vol}{vol}

\makeatletter
\newcommand*{\centerfloat}{%
  \parindent \z@
  \leftskip \z@ \@plus 1fil \@minus \textwidth
  \rightskip\leftskip
  \parfillskip \z@skip}
\makeatother

\begin{document}

\title{Elliptic Curve Multiset Hash}

\ifcompj
\author{Jeremy Maitin-Shepard}
\email{jbms@cs.berkeley.edu}
\affiliation{UC Berkeley, USA}

\author{\\Mehdi Tibouchi}
\affiliation{NTT Secure Platform Laboratories, Japan}

\author{Diego F.\ Aranha}
\affiliation{University of Campinas, Brazil}

\shortauthors{J.~Maitin-Shepard, M.~Tibouchi, D.F.~Aranha}

\else
\author{Jeremy Maitin-Shepard\\
UC Berkeley\\
{\tt\small jbms@cs.berkeley.edu}
\and
Mehdi Tibouchi\\
NTT Secure Platform Laboratories\\
{\tt\small tibouchi.mehdi@lab.ntt.co.jp}
\and
Diego F.\ Aranha\\
Institute of Computing, University of Campinas\\
{\tt\small dfaranha@ic.unicamp.br}}
\date{}
\newtheorem{definition}{Definition}
\newtheorem{theorem}{Theorem}
\maketitle
\fi

\begin{abstract}
A homomorphic, or incremental, multiset hash function, associates a hash value to
arbitrary collections of objects (with possible repetitions) in such a
way that the hash of the union of two collections is easy to compute from
the hashes of the two collections themselves: it is simply their sum
under a suitable group operation. In particular, hash values of large
collections can be computed incrementally and/or in parallel. Homomorphic
hashing is thus a very useful primitive with applications ranging from
database integrity verification to streaming set/multiset comparison and
network coding.

\medskip
Unfortunately, constructions of homomorphic hash functions in the
literature are hampered by two main drawbacks: they tend to be much
longer than usual hash functions at the same security level (e.g.{} to
achieve a collision resistance of $2^{128}$, they are several thousand
bits long, as opposed to $256$ bits for usual hash functions), and they
are also quite slow.

\medskip
In this paper, we introduce the Elliptic Curve Multiset Hash (ECMH),
which combines a usual bit string-valued hash function like BLAKE2 with
an efficient
encoding into binary elliptic curves to overcome both difficulties. On
the one hand, the size of ECMH digests is essentially optimal: $2m$-bit
hash values provide $O(2^m)$ collision resistance. On the other hand, we
demonstrate a highly-efficient software implementation of ECMH, which our thorough empirical evaluation shows to be capable of processing over 3 million set elements per second on a \SI{4}{GHz} Intel
Haswell machine at the 128-bit security level---many times
faster than previous practical methods.

\medskip
While incremental hashing based on elliptic curves has been considered previously~\cite{brown2008encrypted}, the proposed method was less efficient, susceptible to timing attacks, and potentially patent-encumbered~\cite{brown2007method}, and no practical implementation was demonstrated.
\ifcompj\else\par
\fi
\end{abstract}

\ifcompj\maketitle\fi

\section{Introduction}

\subsubsection*{Homomorphic hashing}
A \emph{multiset} is a generalization of a set in which each element has
an associated integer \emph{multiplicity}. Given a possibly infinite set
$A$, a set (resp.{} multiset) homomorphic hash function on $A$ maps
finite subsets of $A$ (resp.{} finitely-supported multisets on $A$) to
fixed-length hash values, allowing incremental updates: when new elements
are added to the (multi)set, the hash value of the modified (multi)set
can be computed in time proportional to the degree of modification.

The incremental update property makes homomorphic hashing a very useful
and versatile primitive. It has found applications in many areas of
computer security and algorithmics, including network
coding~\cite{DBLP:conf/infocom/GkantsidisR06} and verifiable peer-to-peer
content distribution~\cite{DBLP:conf/sp/KrohnFM04}, secure Internet
routing~\cite{DBLP:conf/nsdi/SubramanianRSSK04}, Byzantine fault
tolerance~\cite{DBLP:conf/osdi/CastroL99,DBLP:journals/tocs/CastroL02},
\emph{streaming} set and multiset equality
comparison~\cite{cathalo2009comparing}, and various aspects of database
security, such as access pattern privacy~\cite{DBLP:conf/ccs/2008} and
integrity protection~\cite{clarke2003incremental}.

This latter use case provides a simple example of how the primitive is
used in practice: one can use homomorphic hashing to verify the integrity
of a database with a transaction log, by computing a hash value for each
transaction in such a way that the hash of the complete
database state is equal to the (appropriately-defined) sum of the hashes
of all transactions. Another observation~\cite{bellare1997new} is
that homomorphic hashing can be used for incremental and parallel hashing
of lists, arrays, strings and other similar data structures: for example,
the list $(b_1, \ldots, b_n)$ can be represented as the set
$\Set*{(1,b_1), \ldots, (n,b_n)}$, and it suffices to apply the
homomorphic hash function to that set.

\subsubsection*{Constructing homomorphic hash functions}
A framework for constructing provably secure homomorphic hash functions
(in some suitably
idealized model, such as the random oracle model) was introduced by
Bellare and Micciancio~\cite{bellare1997new}, and later extended to the multiset hash
setting by Clarke~{et al.}~\cite{clarke2003incremental}, and revisited by
Cathalo~{et al.}~\cite{cathalo2009comparing}.

Roughly speaking, the framework of Bellare and Micciancio can be described as
follows. To construct a (multi)set homomorphic hash function on $A$, one
can start with a usual hash function $\hat H$ from $A$ to some additive
group $G$, and extend it to finite subsets of $A$ (resp.{} multisets on
$A$) by setting $H(\Set*{a_1,\dots,a_n}) = \hat H(a_1)+\cdots+\hat
H(a_n)$ (resp.{} $H(\Set*{a_1^{m_1},\dots,a_n^{m_n}}) = m_1\cdot \hat
H(a_1)+\cdots+ m_n\cdot\hat H(a_n)$, where $m_i$ is the multiplicity of
$a_i$).  And in fact, it is clear that all possible homomorphic hash
functions arise in that way.  Note that as in Clarke~{et
al.}~\cite{clarke2003incremental}, and unlike the original framework of
Bellare and Miciancio~\cite{bellare1997new}, there is no block index $i$ included in the hash $\hat{H}(a_i)$ of each element $a_i$, because we are hashing unordered sets/multisets, rather than ordered sequences of blocks.

Assume that the underlying hash function $\hat H$ is ideal (i.e.{} it
behaves like a random oracle). Then we can ask when the corresponding
homomorphic hash function $H$ is secure (collision resistant, say). This
translates to a knapsack-like number-theoretic assumption on the group
$G$, which Bellare and Micciancio show holds, for example, when the
discrete logarithm problem is hard in $G$.

Concretely, Bellare and Micciancio and the authors of subsequent works
propose a number of possible instantiations for $H$ which
essentially amount to choosing $G = \Z_p^\times$ or $G = \Z_m^n$ for
suitable parameters $p,m,n$. These concrete instantiations yield simple
implementations, but they all suffer from suboptimal output size (they
require outputs of several thousand bits to achieve collision resistance
at the $128$ security level), and their efficiency is generally
unsatisfactory. Essentially all practical applications of homomorphic
hashing in the security literature seem to focus on the case $G =
\Z_p^\times$, called \muhash.

\subsubsection*{Our contributions}
Within Bellare and Micciancio's framework, constructing a homomorphic
hash function amounts to choosing a group $G$ where the appropriate
number-theoretic assumption holds, together with a hash function to $G$
whose behavior is close enough to ideal for the security proof to go
through.

In this paper, we propose a novel concrete construction of a multiset
hash function by choosing $G$ as the group of points of a binary elliptic
curves, and picking the hash function following the approach of
Brier~{et al.}~\cite{brier2010efficient} (which we improve upon slightly)
applied to the binary curve variant of Shallue and van de Woestijne's
encoding function~\cite{shallue2006construction}. We also describe a
software implementation of our proposal (building upon the work of
Aranha~{et al.}~\cite{aranha2014binary} for binary curve hashing, and using
BLAKE2~\cite{aumasson2013blake2} as the actual underlying hash function)
and provide extensive performance results showing that our function
outperforms existing methods by a large margin on modern CPU
architectures (especially those supporting carry-less multiplication).
Furthermore, choosing an elliptic curve (with small cofactor) for the
group of hash values solves the ``output size'' problem of homomorphic
hashing outright: $O(2^n)$ collision security is achieved with roughly
$2n$-bit long digests. Yet, they do not seem to have been used in
concrete implementations of homomorphic hashing so far\footnote{One can mention EECH~\cite{brown2008encrypted} as
    relevant related work that also uses binary curves for hashing, but
    the authors didn't consider \emph{homomorphic} hashing at all, and
their functions seems poorly suited for that goal. See
Section~\ref{sec:discussion} for a more detailed discussion.}. One can
wonder why; the most
likely explanation is that usual methods for hashing to elliptic
curves are far too inefficient to make curves attractive from a
performance standpoint: almost all such methods require at least one full
size exponentiation in the base field of the curve, which will be much
more costly by itself than the single multiplication (in a much larger
field) required by \muhash---even on curves over fast prime fields at the
128-bit security level~\cite{aranha2014binary}, such an encoding function is over $3$ times slower than \muhash{} at equivalent
security on Haswell, and over $20$ times slower than our construction.  Only by using binary curves and relatively
sophisticated implementation techniques do we avoid that stumbling block
and prove that elliptic curves can be competitive. As a result,
we achieve a processing speed of over 3 million set elements per second on a
\SI{4}{GHz} Intel Haswell CPU at the $128$-bit security level.
Speedups are expected with the release of Intel Broadwell processors and
its improved implementation of carry-less multiplication.

\subsubsection*{Are binary elliptic curves safe?}

Recently, new developments have been announced regarding the asymptotic
complexity of the discrete logarithm problem on binary elliptic curves,
particularly by Semaev~\cite{semaev2015}. These results are somewhat
controversial, since they are based on heuristic assumptions that
prevailing evidence suggests are unlikely to
hold~\cite{kosters2015,DBLP:conf/crypto/HuangKY15}, and
their storage requirements appear to make them purely theoretical
anyway~\cite{galbraith2015ellipticnews}.

However, if Semaev's claims of an $L[1/2]$ attack turn out to be correct,
the asymptotic security of binary elliptic curve-based ECMH would be
reduced. The \emph{concrete} security of our construction, on the other
hand, would be completely unaffected on curves of up to 300+ bits (and in
particular at the 128-bit security level on GLS254), since the claimed
attack is worse than generic attacks on such curves. Moreover, even if
\emph{actually practical} $L[1/2]$ attacks were to be found, ECMH on
binary curves is likely to remain attractive, since it mainly competes
against \muhash, which is vulnerable to an $L[1/3]$ subexponential
attack.

For all these reasons, we believe that ECMH on binary elliptic curves is
a safe choice for security-minded practitioners, and that the switch from
\muhash{} to ECMH is entirely justified in view of the considerable
performance gain (which lets designers choose a higher security margin
and still come out far ahead).

\section{Homomorphic Multiset Hash Function}

Formally, we define a multiset $M \in \Z^{(A)}$ as a function with finite
support mapping a base set $A$ to the integers $\Z$.  As an extension of
the usual definition in which multiplicities are restricted to $\Z_{\geq
0}$, we allow negative multiplicities as well.  We will implicitly consider
subsets $S \subseteq A$ to be multisets in $\Z^{(A)}$.

Clarke~{et al.}~\cite{clarke2003incremental} introduce a definition of a multiset hash function that efficiently supports incrementally adding (multisets of) elements.  We give a simpler (but nearly equivalent\footnote{We give a proof of equivalence (under a mild assumption) in~\cref{sec:clarke-equivalence}.}) definition that makes the connection to homomorphic hash functions~\cite{krohn2004fly} explicit:

\begin{definition}
  Let $A$ be a set and let $(G, +_G)$ be a finite group.  A function $H
\colon \Z^{(A)} \rightarrow G$ that maps multisets over the base set $A$
to a point in $G$ is said to be a \emph{homomorphic multiset hash
function} if $H$ is a group homomorphism from the pointwise-additive
group of functions $(\Z^{(A)}, +)$ to $(G, +_G)$; equivalently,
  $H(M_1 + M_2) = H(M_1) +_G H(M_2)$ for all $M_1, M_2 \in \Z^{(A)}$.  We define $\hat{H} \colon A \rightarrow G$ by $\hat{H}(a) = H(\Set{a})$.
  \label{def:homomorphic-multiset-hash}
\end{definition}

This definition minimally captures an intuitive notion of a multiset hash
function that supports incrementally adding and removing (multisets of)
elements.  These incremental updates are efficient assuming that addition
and negation in $G$ can be performed efficiently and $H(M)$ can be
computed efficiently (e.g.\ in time linear in the representation length
of the non-zero values of $M$).  Note that since pointwise addition in
$\Z^{(A)}$ is commutative, the relevant subgroup $H(\Z^{(A)}) \leq G$ is necessarily commutative, and therefore without loss of generality we can assume that $G$ is commutative.  It may seem that it is too strong of an assumption to require a group structure on $G$, or equivalently, that (multisets of) elements can be removed as well as added.  In fact, provided that $+_G$ is lossless, in that $a +_G b = a +_G c$ implies $b = c$, there is no loss of generality.  We show in \cref{sec:monoid-to-group} that we can construct a group that supports (efficient) incremental removals based only on (efficient) incremental additions.

Since the set of singleton subsets of $A$ generates the group $(\Z^{(A)}, +)$, $H$ can conversely be uniquely defined by $\hat{H}$:
\[ H(M) = \sum_{a \in A} M(a) \cdot \hat{H}(a). \]
Indeed, this is precisely the \emph{randomize-then-combine paradigm}
proposed by Bellare and Micciancio~\cite{bellare1997new} for incremental
hashing of messages, which is readily (in fact, more naturally than to
message hashing) applied by Clarke~{et al.}~\cite{clarke2003incremental} to multiset hashing.  Our goal is to minimize the computational cost for computing $H$ and the representation size for elements of $G$ while achieving a given level of collision resistance.

\subsubsection*{Collision resistance}

A \emph{collision} for a hash function $H$ is a pair $x, x'$ such that $x \not= x'$ but $H(x) = H(x')$.  For any group-homomorphic hash function $H$ from a group $(X, +)$ to $(G, +)$, a collision can equivalently be defined as a value $x \in \ker H \setminus \Set{0_X}$.   By the birthday bound that applies to any hash function, a collision can be found with at most expected $O(\sqrt{\card{G}})$ hash computations; we can hope to design a multiset hash function for which expected time $\Omega(\sqrt{\card{G}})$ is also a lower bound.\footnote{In this and the other collision bounds that follow, it is assumed that the expectations are taken over a random choice of hash function $H$ and group $(G, +_G)$ from some hash function family (distribution) $\mathcal{H}$.}

A \emph{preimage attack} seeks to invert the hash function, namely to find a value $x$ such that $H(x) = y$, for a random element $y$ in the image of $H$.  We can hope to design a multiset hash function for which the expected time complexity of the best preimage attack is also equal to the generic upper bound $\Omega(\card{G})$.  Note that for a homomorphic hash function we do not consider preimage attacks on the identity element $0_G$, since its preimage is fixed.

A \emph{second preimage attack} seeks to find a value $x'$ such that $H(x') = H(x)$, for some known value $x$.  Since a second preimage implies a collision, the time complexity of a second preimage attack is lower bounded by the time complexity of the best  collision attack, ideally $\Omega(\sqrt{\card{G}})$.  For a general, non-homomorphic hash function,  we can hope that the best attack has expected $\Omega(\card{G})$ time complexity.  For any homomorphic hash function, however, the group structure implies that a second preimage attack is no harder than a collision attack (with expected time complexity upper-bounded by $O(\sqrt{\card{G}})$).

\section{Generic multiset hash families}

A random oracle $\hat{H} \colon A \rightarrow G$ clearly achieves the
optimal preimage resistance of $\Theta(\card{G})$ and the optimal
collision resistance of $\Theta(\sqrt{\card{G}})$, in the sense that at
least this many oracle queries are needed to compute preimages and
collisions respectively.

It does not follow, however, that the associated multiset hash function
$H = H_G\colon \Z^{(A)} \rightarrow G$ has the same security level; for
example, if we choose $G=\Z_2^n$, then $O(n)$ oracle queries, instead of
$\Omega(2^n)$, are enough to find arbitrary preimages in polynomial time
by solving a simple $n\times n$ linear system over $\Z_2$. However,
Bellare and Micciancio~\cite{bellare1997new} have shown (in the set hash setting, but
this generalizes naturally to multisets) how to obtain a security
reduction for $H_G$ based on a computational hardness assumption on the
group $G$. For concrete choices of $G$, that hardness assumption is
related to standard number theoretic problems, such as the discrete
logarithm problem or modular knapsacks.

When $G=\Z_p^\times$, the resulting multiset hash function $H_G$ is
essentially \msetmuhash{}~\cite{clarke2003incremental}, the multiset
variant of \muhash{}~\cite{bellare1997new}. When $G=\Z_m^n$, we
essentially obtain \msetvaddhash{}~\cite{clarke2003incremental}, the
multiset variant of \lthash{} (for $n > 1$) or \adhash{} (for $n =
1$)~\cite{bellare1997new}. These functions all have security reductions
in the framework sketched above.

It is relatively easy to find plausible concrete instantiations of the
random oracle $\hat H$ to a group like $\Z_2^n$, but for more general groups, this is
usually more complicated, and as a result it is often convenient to
replace $\hat H$ by a \emph{pseudo-random oracle}, i.e.\ a construction
that is indifferentiable from a random oracle in the sense of
Maurer~{et al.}~\cite{maurer2004indifferentiability}. Typically, we can take
$\hat H$ of the form $\hat H(a) = f\big(h(a)\big)$ where $h\colon A\to X$
is a random oracle to some intermediate set $X$ (such as bit strings, so
that we can plausibly instantiate it with standard hash function
constructions like SHA-2\footnote{We will assume that elements of $A$ can
be readily encoded as octet strings.}) and $f\colon X \rightarrow G$ is
an \emph{admissible encoding
function}~\cite{boneh2001identity,brier2010efficient} that has the
property of mapping the uniform distribution over $X$ to a distribution
indistinguishable from uniform over $G$.

\subsubsection*{Security bounds}
\label{sec:generalized-multiset-hash-security-bounds}

\adhash{} is appealing for its simplicity, but is far from optimal in terms of hash code size.  In the \emph{set} hashing setting (i.e.\ $M(a) \in \Set{0,1}$), the best known attack is the generalized birthday attack~\cite{wagner2002generalized}; under the assumption that this attack is optimal, the group $\Z_{2^n}$ corresponds to a security level of roughly $2 \sqrt{n}$ bits.  In the \emph{multiset} hashing setting, \adhash{} is completely impractical due to the extremely large hash code sizes $n$ required to defeat lattice reduction attacks described in \cref{sec:adhash-attack}.

There are reductions from computing discrete logarithms in a group $G$ to
finding collisions in the corresponding random oracle multiset hash
function
$H_{G}$~\cite{impagliazzo1996efficient,bellare1997new,clarke2003incremental}.
These reductions can be used to prove a collision resistance property for
the generic multiset hash family over any group in which computing
discrete logarithms is hard, such as $\Z_p^\times$.  However, because
discrete logarithms in $\Z_p^\times$  can be solved by e.g.\ the Number
Field Sieve with (heuristic) subexponential time complexity $L_p\big[1/3,
\sqrt[3]{64/9}\big]$~\cite[p.~128]{menezes2010handbook}, it is usually
estimated that we need to choose $p$ of around $3200$ bits for $128$-bit
security (see for example the evaluation of the ECRYPT~II report on key
sizes~\cite{smart2010ecryptkeysizes}).  In contrast, in a generic group, discrete logarithms
cannot be computed faster than expected time $\Theta(\sqrt{\card{G}})$,
which is also the optimal collision resistance.



\section{Elliptic Curve Multiset Hash}
\label{sec:elliptic-curve-multiset-hash}

For properly chosen elliptic curves over finite fields,
there are no known algorithms for solving the discrete logarithm problem in the elliptic curve group faster than in a generic group, i.e.\ expected time $\Theta(\sqrt{\card{G}})$.  Therefore, there is a clear possibility for using an elliptic curve group to obtain a given level of collision resistance with a much lower group size than with \msetmuhash{}.


Applying the generic multiset hash construction to elliptic curve groups
presents a problem, however: while it is easy to define a very efficient
admissible encoding from $\{0,1\}^k$ to $\Z_p^\times$ for sufficiently
large $k$, an admissible encoding to an elliptic curve group is not so
easily defined.  While constructions for admissible encoding functions
have been demonstrated~\cite{brier2010efficient}, their computational
cost is higher than we would like.

\subsection{Generalized discrete logarithm security reduction}

In fact, we can significantly relax the requirement on the encoding
function $f$ and still obtain a very tight reduction, due to random
self-reducibility of the discrete logarithm problem.  Our relaxed
requirement is related to the definition of $\alpha$-weak encodings by
Brier~{et al.}~\cite{brier2010efficient}, and is satisfied in practice by a large class of encoding functions~\cite{brier2010efficient}.
\begin{definition}
  A function $f \colon S \rightarrow R$ between finite sets is said to be an $(\alpha, \beta)$-weak encoding, for integer $\alpha \geq 1$ and real value $\beta \geq 1$, if it satisfies the following properties:
\begin{enumerate}
\item \emph{Samplable}: there is an efficient randomized algorithm for computing $\card{f^{-1}(r)}$ and sampling uniformly from $f^{-1}(r)$ for any $r \in R$.
\item $\card{f^{-1}(r)} \leq \alpha$ for all $r \in R$.
\item $\Expect_r[ \card { f^{-1}(r) } / \alpha ] \geq 1 / \beta$.
\end{enumerate}
\end{definition}
An $(\alpha, \beta)$-weak encoding function $f$ allows us to efficiently
sample $s \in S$ uniformly at random using $\beta$ uniform samples $r \in
R$ in expectation, with the property that $f(s) = r$ for any accepted
sample $s$ obtained from $r$.\footnote{Under the definition of Brier~{et
al.}~\cite{brier2010efficient}, an $\alpha$-weak encoding $f$ is an $(\alpha \card{S} /  \card{R}, \alpha^2 \card{S} / \card{R})$-weak encoding.  Our definition allows for a tighter bound to be given in \cref{thm:dlpreduction}.}

\begin{definition}
\label{def:alpha-beta-weak-multiset-hash-family}
Let $f \colon X \rightarrow G$ be an $(\alpha, \beta)$-weak encoding from
$X$ to the abelian group $G$. Assume that $G$ admits as a direct factor a cyclic
subgroup $\langle g \rangle$ of prime order
$\rho$, and that we can efficiently sample from the
complement group $\overline{\langle g \rangle}$ in the direct factor
decomposition $H=\langle g\rangle \oplus \overline{\langle g \rangle}$. Given a random oracle
$h\colon A\to X$, we denote by $\hat H_f$ the function $A\to G$ given by $\hat
H_f(a) = f\big(h(a)\big)$, and by $H_f\colon \Z^{(A)}\to G$ the
associated multiset hash function.
\end{definition}

The following theorem shows that finding a collision in $H_f$
with multiplicities up to $\rho - 1$ is as hard as computing discrete
logarithms to the base $g$, up to a small factor that depends on $\beta$.
Note that $H_f$ does not depend on the choice of subgroup
$\langle g \rangle$, but the strongest security result is obtained by
choosing the largest prime-order subgroup.  The requirement of efficient
samplability of $\overline{\langle g \rangle}$ is easily satisfied in
practice, since efficiency concerns regarding representation size dictate
that $\overline{\langle g \rangle}$ be as small as possible (usually
having at most $8$ elements, and most of the time only $1$ or $2$).

\begin{restatable}{theorem}{dlpreduction}
\label{thm:dlpreduction}
Let $H_f$ be a multiset hash function as in
\cref{def:alpha-beta-weak-multiset-hash-family}.  Given an algorithm
$\mathcal{C}$ with access to the underlying random oracle $h$ that finds
a non-empty multiset $M \in \ker H_f$ with $\lvert M \rvert_\infty <
\rho$, in expected time $t'$ with probability $\epsilon'$
using $q$ queries to $h$, discrete logarithms to the
base $g$ can be computed with probability $\epsilon = \epsilon'/2$ in
expected time $t + T_1 + q T_2 + q \beta T_3 + L T_4$, where $L \geq
\lvert M \rvert_0$ is a bound on the length of the output of
$\mathcal{C}$; $T_1, \dotsc, T_4$ denote the time required for a constant
number of group operations, and are given in the proof.
\end{restatable}
\begin{proof}
See Appendix~\ref{sec:proofthmdlpreduction}.
\end{proof}

Concretely, if $G = E(\F_{p^m})$ is the group of $\F_{p^m}$-rational
points on a suitable elliptic curve $E$ chosen to avoid any discrete
logarithm weaknesses, with a subgroup $\langle g \rangle$ of prime order
$\rho \geq \card{G}/4$, and $f$ is an $(\alpha, \beta)$-weak encoding
function with small constant $\beta$, then $H_f$ has collision
resistance roughly $p^{m/2}/2$.  Since an element of $E(\F_{p^m})$ can be
represented using $\lceil \log_2 p^m \rceil$ bits, the collision
resistance of $H_f$ is essentially optimal (to within a few
bits).

\subsection{Shallue-van de Woestijne (SW) encoding in characteristic 2}

The Shallue-van de Woestijne (SW) algorithm for characteristic 2 fields~\cite{shallue2006construction} can be used to map any point $w \in \F_{2^m}$ to a pair $(x,y) \in \left(\F_{2^n}\right)^2$ satisfying an arbitrary elliptic curve equation
\begin{equation}
  E_{a,b} \colon y^2 + x \cdot y = x^3 + a \cdot x^2 + b,
  \label{eq:elliptic-curve}
\end{equation}
where $a, b \in \F_{2^m}$.  It constructs three values of $x$ from $w$ with the property that at least one necessarily has a corresponding value $y$ satisfying \cref{eq:elliptic-curve}.  In addition to the usual arithmetic operations over $\F_{2^m}$, its definition depends on three \emph{linear} maps:
\begin{enumerate}
\item the \emph{trace} function $\mathrm{Tr} \colon \F_{2^m} \rightarrow \F_2$ defined by $\mathrm{Tr}(x) = \sum_{i=0}^{m-1} x^{2^i}$;~\cite[p.~130]{hankerson2004guide}
\item a \emph{quadratic solver} function $\mathrm{QS} \colon \Set{ x \in \F_{2^m} \given \mathrm{Tr}(x) = 0 } \rightarrow \F_{2^m}$ that satisfies $\mathrm{QS}(x)^2 + \mathrm{QS}(x) = x$ and $\mathrm{QS}(0) = 0$;
\item $\mathrm{coeff}_0 \colon \F_{2^m} \rightarrow \F_2$ where $\mathrm{coeff}_0(x)$ is the zeroth coefficient of any (fixed) polynomial representation of $x$.
\end{enumerate}

An optimized version of the algorithm that requires only a single field inversion~\cite{aranha2014binary} is shown as \cref{alg:sw-encode}.  The algorithm is parameterized by a value $t \in \F_{2^m}$ satisfying $t^4 + t \not= 0$; for fields of degree $m > 4$, we can choose $t = z$ where $z$ is the indeterminate in the polynomial representation of $\F_{2^m}$.  The result $(x, \lambda) = (x, x + y / x)$ is represented in $\lambda$-affine coordinates~\cite{oliveira2014two} for efficiency.\footnote{There is exactly one point with $x = 0$ satisfying \cref{eq:elliptic-curve}: $(x = 0, y = \sqrt{b})$.  When using $\lambda$-affine coordinates, this value must be represented specially.}  The addition to $\lambda$ of $\mathrm{coeff}_0(w)$ in \cref{alg:sw-encode:coeff-line} is not part of the original SW algorithm; this trivial addition serves to halve the number of collisions at essentially no extra cost.

\begin{algorithm*}
  \caption{Optimized Shallue--van de Woestijne encoding in characteristic 2~\cite{aranha2014binary}}
  \label{alg:sw-encode}
  \begin{algorithmic}[1]
    \Require $t \in \F_{2^m}$ such that $t \cdot (t + 1) \cdot (t^2 + t + 1) = t^4 + t \not= 0$
    \item[\textbf{Precompute:}] $t_1 = \frac{t}{t^2 + t + 1}$, $t_2 = \frac{1 + t}{t^2 + t + 1}, t_3 = \frac{t \cdot (1 + t)}{t^2 + t + 1}; t_j^{-1} = 1 / t_j$ for $j = 1, 2, 3$
    \Function{SWChar2}{$w \in \F_{2^m}$ ; $a, b \in \F_{2^m}$}
      \State $c\gets w^2 + w + a$
      \If{$c = 0$} \Comment{This condition may hold only if $\mathrm{Tr}(a) = 0$}
        \State \Return $(x = 0, y = \sqrt{b})$ \Comment{This is the single point satisfying \cref{eq:elliptic-curve} with $x = 0$}
      \EndIf
      \State $c^{-1} \gets 1 / c$
      \For{$j = 1$ to $3$}
        \State $x \gets t_j \cdot c$
        \State $x^{-1} \gets t_j^{-1} \cdot c^{-1}$
        \State $h_j \gets (x^{-1})^2 \cdot b + x + a$
        \If{$\mathrm{Tr}(h_j) = 0$} \Comment{This condition necessarily holds for at least one $j$}
          \State $\lambda \gets \mathrm{QS}(h_j) + x + \mathrm{coeff}_0(w)$  \Comment{$c$ does not depend on $\mathrm{coeff}_0(w)$}
          \label{alg:sw-encode:coeff-line}
          \State \Return $(x, \lambda)$ \Comment{$y = (\lambda + x) \cdot x = \mathrm{QS}(h_j) \cdot x$}
        \EndIf
      \EndFor
    \EndFunction
  \end{algorithmic}
\end{algorithm*}

It is clear from the definition that the number of preimages of any point
$(x, \lambda)$ under \textproc{SWChar2} is at most $\alpha = 3$, since $c
\in \Set{ t_j^{-1} \cdot x \given j = 1, 2, 3 }$ and $w$ is uniquely
determined from $c$, $x$, and $\lambda$ by
\begin{align*}
  w &\in \Set{ \mathrm{QS}(c - a), \mathrm{QS}(c - a) + 1 }, \\
  \mathrm{coeff}_0(w) &= \lambda + x + \mathrm{QS}(x^{-2} \cdot b + x + a).
\end{align*}
The preimage set for any point $(x, \lambda)$ can be efficiently computed
by these same formulas.  Furthermore, Aranha~{et al.}~\cite{aranha2014binary}
show that the proportion of curve points with $k$ preimages under $\textproc{SWChar2}$ for $k =
0,1,2,3$ is $9/32$, $15/32$, $7/32$, and $1/32$, respectively, up to an
error term of $O(2^{-n/2})$.  It follows that
$\Expect_P[\card{\textproc{SWChar2}^{-1}(P)} / \alpha] = 1/3 \pm
O(2^{-n/2})$, and therefore, \textsc{SWChar2} is an $(\alpha,
\beta)$-weak encoding with $\beta = 3 + O(2^{-n/2})$.

\subsection{Hash function definition}

Based on this encoding function, we define the \emph{elliptic curve multiset hash
(ECMH)}: given a binary elliptic curve group $E_{a,b}(\F_{2^m})$ and an
intermediate hash function $h \colon A \to \Z_{2^m}$ (modeled as a random
oracle), we define $\mathrm{ECMH}_{a,b,h}(x) = \textsc{SWChar2}(h(x); a, b)$.  Commonly used elliptic curves over $\F_{2^m}$, including the NIST-recommended ones,
have a generator of prime order $\rho > 2^{m-2}$ with an easily determined complement group of size $h = \card{\overline{\langle g \rangle}} \leq 4$.  Thus, the samplability requirement on $\overline{\langle g \rangle}$ is easily satisfied in practice.  Hence, by \cref{thm:dlpreduction}, finding a collision in ECMH is as hard (up to a small constant factor) as computing discrete logarithms to the base $g$, which we assume to be $O(2^{m/2})$.

Similar suitable encoding algorithms exist for elliptic curves over fields of characteristic $p > 2$~\cite{shallue2006construction,brier2010efficient}, and could also be used to define an elliptic curve multiset hash.  However, the use of a characteristic 2 field eliminates the need for an expensive field exponentiation in order to solve a quadratic equation, which would otherwise dominate the computation time, and on modern CPUs that support fast carry-less multiplication, fast implementations of all other required field operations are also possible for characteristic $2$~\cite{taverne2011jcen}.

\subsection{Compressed representation of curve points}
\label{sec:compressed-point-representation}

The group of $\F_{2^m}$-rational points on an elliptic curve $E_{a,b}$ has order $\card{E_{a,b}(\F_{2^m})} \approx 2^m$.  Each point is naturally represented as a pair $(x,y) \in \F_{2^m}^2$ (or $(x, \lambda) \in \F_{2^m}^2$), but there is a well-known method for encoding a point using just $m + 1$ bits: given $x$ there are at most two possible values for $y$ (or $\lambda$) if $(x,y)$ (or $(x,\lambda)$) satisfy $E_{a,b}$, and they can be recovered efficiently using a small number of field operations.  Thus, a point can be encoded by its $x$ value and a single additional bit to disambiguate the two possible points.  The elliptic curve group identity element (the point at infinity) can be encoded specially without increasing the representation size, by using a bit sequence that would not otherwise encode a valid point.

\section{Implementation}

We developed an optimized implementation of elliptic curve multiset hash (ECMH) as an open-source C++ library~\cite{maitinshepard2015ecmhLibrary},
with support for all NIST-recommended binary elliptic curves~\cite{fips186-4} and the record-breaking GLS254 curve~\cite{oliveira2014two}, as well as several other SEC~2-recommended curves~\cite{sec2-1}.  Using a combination of C++ templates and code generation, we were able to write generic code to support many different configurations without sacrificing runtime performance; only for modular reduction was a custom implementation required for each supported field.  We incorporated existing fast x86/x86-64 polynomial multiplication, squaring, and modular reduction routines for $\F_{2^{163}}$, $\F_{2^{193}}$, $\F_{2^{233}}$, $\F_{2^{239}}$, $\F_{2^{283}}$, $\F_{2^{409}}$, $\F_{2^{571}}$~\cite{bluhm2013fast} and for $\F_{2^{127}}$~\cite{oliveira2014two}.

We implemented field inversion using a polynomial-basis Itoh--Tsujii inversion method making use of multi-squaring tables~\cite{guajardo2002itoh,taverne2011jcen,oliveira2014two,aranha2014binary}.  We generated field inversion routines for each field degree automatically based on an A\textsuperscript{*} search procedure for computing the optimal Itoh--Tsujii addition chain and set of multi-squaring tables, based on a machine-specific cost model estimated from field operation performance measurements~\cite{cryptoeprint:2015:028}.

We also developed optimized implementations of the \msetmuhash{} and
\msetaddhash{} hash functions, based on the modular arithmetic functions in the OpenSSL library version 1.0.1i, for the purpose of comparison.

\subsection{Intermediate hash function}

ECMH requires an intermediate hash function $h \colon A \to \Z_{2^m}$.
Under our assumption that the base set $A$ is the set of octet strings,
we simply require a standard cryptographic hash function (modeled as a
random oracle) with output size $m$.  Given the inherent property of any homomorphic hash function that a single collision leads to arbitrary second preimages, we advise using a keyed hash function when possible to minimize risk.

Any standard hash function with fixed output size greater than $m$ bits can simply be truncated to $m$ bits.  Standard expansion techniques can be used to efficiently generate an arbitrary length $m$ output from a hash function with fixed output size $b < m$.  Sponge constructions, such as Keccak~\cite{bertoni2009keccak}, are particularly convenient since they support arbitrary output sizes.

Both \adhash{} and \muhash{} similarly require intermediate hash functions, but with much larger output sizes $m$ for equivalent security levels.

We designed our implementation to support arbitrary hash functions, but for our performance evaluation, we selected BLAKE2~\cite{aumasson2013blake2} because of its state-of-the-art performance.  For $m \leq 256$, we used the BLAKE2s variant (256-bit output), truncating the output to $m$ bits.  For $256 < m \leq 512$, we used the BLAKE2b variant (512-bit output) with truncation.  For $m > 512$, we used BLAKE2b repeatedly to generate sufficient output, in such a way that the underlying compression function is called a minimum number of times.

\subsection{Linear field operations}
\label{sec:impl:linear-field-operations}

Several key operations for $\F_{2^m}$, such as squaring, multi-squaring ($x \mapsto x^{2^i}$), square root, and half-trace, are linear in the coefficients.  For multi-squaring (useful for inversion) and half-trace, an implementation based on a lookup table can be significantly faster than direct computation~\cite{bos2010ecc2k,taverne2011jcen,oliveira2014two,aranha2014binary}.  The coefficients are split into $\lceil m / \beta \rceil$ blocks of $\beta$ bits, and a separate table of $2^{\beta}$ entries is precomputed for each block position, using a total of $\tablememory{m,\beta} = \lceil m / \beta \rceil \cdot 2^{\beta} \cdot \lceil m / W \rceil \cdot W / 8$ bytes of memory, where $W$ is the word size in bits.  The linear transform can then be computed from the precomputed tables with $k = \lceil m / \beta \cdot \lceil m / W \rceil$ memory accesses and $k - 1$ XOR operations.

\subsection{Blinding for side-channel resistance}

The fastest implementation of ECMH is susceptible to timing and cache side-channel attacks, due to the use of lookup tables (for inversion and QS), and the use of branching (for \textsc{SWChar2}).  A branch-free implementation of \textsc{SWChar2} adds only a few additional multiplications and squarings.  Lookup tables are unavoidable for good performance, but we can blind inversion at a cost of just two multiplications and generation of one random field element.  We likewise can blind QS at a cost of 1 squaring, 2 additions, and generation of one random field element, as well as a few bit operations to ensure the random element is in the image of QS.  In this way we can fully protect against timing and cache side-channel attacks at only a small additional cost.

\subsection{Quadratic extension field}

For even $m$, representing $\F_{2^m}$ as a quadratic extension of $\F_{2^{m/2}}$ results in significantly faster field operations relative to an odd-degree field of roughly the same size: inversion in the extension field requires only one inversion in the base field (effectively reducing the memory and computation costs by nearly a factor of 4 for a table-based multi-squaring implementation), and half-trace requires only 2 half-trace computations in the base field (reducing, for a table-based implementation, the computation cost by a factor of 2 and the memory requirement by a factor of 4)~\cite{oliveira2014two}.  We use this representation to support the GLS254 elliptic curve over $\F_{2^{254}}$~\cite{oliveira2014two}.


\subsection{In-memory representation of elliptic curve points}

Although an element in the elliptic curve group of points $E_{a,b}(\F_{2^m})$ can be represented directly using the standard affine $(x,y)$-representation or the $\lambda$-affine $(x,\lambda)$ representation, and more compactly using just $m+1$ bits as described in \cref{sec:compressed-point-representation}, we can more efficiently perform group operations using the $\lambda$-projective representation $(\tilde{x},\tilde{\lambda},z)$ corresponding to the $\lambda$-affine representation $(x = \tilde{x}/z, \lambda = \tilde{\lambda}/z)$: This representation allows point addition and point doubling to be performed without any field inversions~\cite{oliveira2014two}.


\subsection{Batch \textsc{SWChar2} computation}
\label{sec:batch-swchar2}
A large fraction of the computational cost of our elliptic curve multiset hash construction is due to the single field inversion required by the \textsc{SWChar2} encoding function.  Using Montgomery's trick, $n$ independent elements can be inverted simultaneously at the cost of just 1 field inversion and $3(n-1)$ field multiplications~\cite{shacham2001improving}.  Since field inversion is much more than 3 times as expensive as field multiplication, this provides significant computational savings.


\subsection{Montgomery domain for \msetmuhash{}}
\label{sec:muhash-montgomery}
A key cost in a na\"{\i}ve implementation of \msetmuhash{} is the reduction modulo $p$ required by multiplication in $\Z_p^\times$.  To avoid this cost, we can use the \emph{Montgomery reduction}~\cite{montgomery1985modular} defined by
\ifcompj
\begin{multline*}
  \mathrm{Redc}(t ; p, r) = t \cdot r^{-1} \bmod p, \\ \text{$0 \leq t < p \cdot r$, $r > p$, $\gcd(r,p) = 1$.}
\end{multline*}
\else
\[
  \mathrm{Redc}(t ; p, r) = t \cdot r^{-1} \bmod p, \qquad \text{$0 \leq t < p \cdot r$, $r > p$, $\gcd(r,p) = 1$.}
\]
\fi
If $r$ is chosen to be a power of $2$, or a power of $2^w$, where $w$ is the word size, then the computational cost of $\mathrm{Redc}$ is significantly lower than a reduction $\bmod p$.

We represent an element $x \in \Z_p^\times$ as a triplet $(w, y, z) \in \Z_{p-1} \times \Z_P^\times \times \Z_p^\times$ corresponding to $y / z \cdot r^w \bmod p$, where $r$ is the Montgomery reduction constant.  Multiplication under this representation is defined by
\begin{align*}
(w_1,y_1,z_1)\cdot(w_2,y_2,z_2) &=
  (w_1+w_2, \mathrm{Redc}(y_1y_2), \mathrm{Redc}(z_1z_2)); \\
(w_1,y_1,z_1)\cdot(w_2,y_2,1)   &=
  (w_1+w_2+1, \mathrm{Redc}(y_1y_2), z_1); \\
(w_1,y_1,z_1)\cdot(w_2,1,z_2) &=
  (w_1+w_2-1, y_1, \mathrm{Redc}(z_1z_2)).
\end{align*}

\section{Performance measurement}

As our test platforms we used an Intel Westmere i7-970 \SI{3.2}{\giga\hertz} CPU (with \SI{12}{MiB} L3 cache) and an Intel Haswell i7-4790K \SI{4.0}{GHz} CPU (with \SI{8}{MiB} L3 cache).  Both of these processors support the \cpuinstruction{PCLMULQDQ} instruction for carry-less multiplication, Westmere being the first Intel architecture to support it; on the much more recent Haswell architecture, where this instruction has significantly lower cost, alternative modular reduction routines based on it are used for $\gf{163}$, $\gf{283}$, and $\gf{571}$ for a modest gain in performance~\cite{bluhm2013fast}.    Our implementation used a word size of $W = 128$ bits and a block size of $B = 8$ bits for all half trace and multi-squaring tables.  All code was compiled separately for each architecture using version 3.5 of the Clang compiler at the highest optimization level.

\subsection{Robust operation timing}
\label{sec:execution-time-measurement}

We measured the execution time of all operations in CPU cycles, using the combination of \cpuinstruction{RDTSC}, \cpuinstruction{RDTSCP}, and \cpuinstruction{CPUID} instructions recommended by Intel~\cite{paoloni2010benchmark}.  To improve accuracy and reduce variance, we disabled TurboBoost, frequency scaling, and HyperThreading, and ensured that a single non-boot CPU core was used for all benchmarks on each machine.  For each operation, we estimated the benchmarking overhead and subtracted it from the measured number of cycles.  Additionally, we automatically determined a per-measurement repeat count for each operation that ensured the benchmarking overhead was less than 10\%.

The execution time was computed as the median of the cycle measurements; the number of cycle measurements for each operation from which the median was computed was at least 1000 and chosen automatically to ensure a sufficiently small 99\% confidence interval on the median estimate (less than the larger of $\sfrac{1}{1000}$ of the estimated median or $\sfrac{1}{10}$ of a cycle).  For consistency, we ensured warm-cache conditions for all estimates by discarding the first 2000 measurements.

\subsection{Consistent measurement of memory-dependent operations}
\label{sec:memory-dependent-measurement}
For operations with data-dependent memory accesses, such as table-based multi-squaring, half trace computation, and the higher-level operations based on these primitives, we measured the aggregate execution time for a set of inputs guaranteed to induce a uniform memory access pattern (and then divided by the number of inputs), in order to obtain worst-case warm-cache estimates.  Failure to do so results in a large underestimate of execution time.

We also observed the performance characteristics of table operations to be significantly affected by the size of the virtual memory pages backing the tables; in particular, on the x86-64 test machines, both the base level performance and the scaling of execution times with increasing table size were significantly better with \SI{2}{MiB} (huge) pages than with \SI{4}{KiB} pages, due to the cost of translation lookaside buffer (TLB) misses.  The Linux transparent huge page support (introduced in Linux version 2.6.38) results in some, but not all, memory regions being backed automatically by huge pages, depending on a number of factors including region alignment and physical memory fragmentation; when not taken into account, this significantly reduced the reliability of our performance measurements.  For consistent performance, we therefore ensured that all lookup tables were backed by huge pages.


\section{Results}
\label{sec:performance-results}

In order to obtain performance results for a full range of security levels, we evaluated the performance of ECMH using each of the following eight elliptic curves: sect163k1~\cite{sec2-1} (NIST K-163~\cite{fips186-4}), sect193r1~\cite{sec2-1}, sect233k1~\cite{sec2-1} (NIST K-233~\cite{fips186-4}), sect239k1~\cite{sec2-1}, GLS254~\cite{oliveira2014two}, sect283k1~\cite{sec2-1} (NIST K-283~\cite{fips186-4}), sect409k1~\cite{sec2-1} (NIST K-409~\cite{fips186-4}), and sect571k1~\cite{sec2-1} (NIST K-571~\cite{fips186-4}).

\begin{figure}
  \centering
  \includegraphics{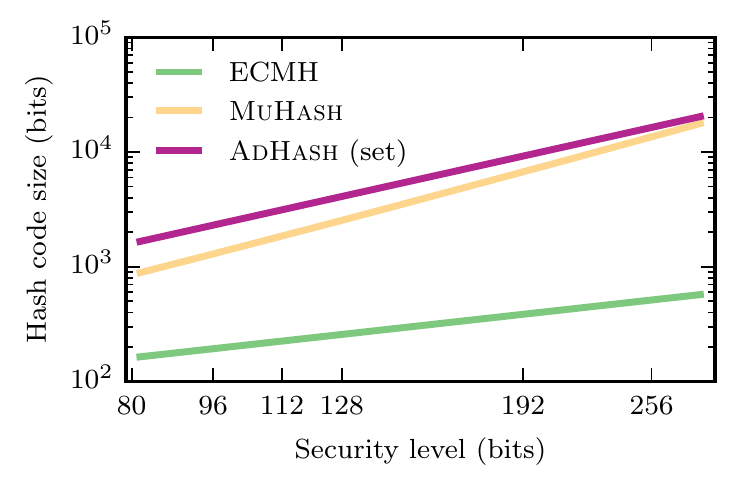}
  \caption{Security level attained as a function of hash code size for Elliptic Curve Multiset Hash (ECMH), \adhash{} (restricted to \emph{set} hashing), and \muhash{}.  For \muhash{}, the \emph{multiset} hashing security level was determined based on the conjectured time-complexity $L_p\left[1/3, \sqrt[3]{64/9}\right]$ of the Number Field Sieve for solving discrete logarithms in $\Z_p^\times$~\cite[p.~128]{menezes2010handbook}.  For \adhash{}, we determined a \emph{set} hashing security level of $2 \sqrt{n}$ corresponded to groups $\Z_{2^n}$ based on the assumption that the generalized birthday attack~\cite{wagner2002generalized} is optimal.}
  \label{fig:hash_code_size}
\end{figure}

Based on \cref{thm:dlpreduction} and the assumed hardness of the Elliptic Curve Discrete Logarithm Problem, the ECMH using an elliptic curve group of order $\rho$ has collision resistance of $O(\sqrt{\rho})$, corresponding to a security level $\log_2 \rho$ bits.  We also evaluated \muhash{} and \adhash{} (for set hashing only) using group sizes corresponding to the same range of security levels.  The correspondence between security level and hash code size under each method is shown in \cref{fig:hash_code_size}.

\begin{figure*}
  \centering
  \includegraphics{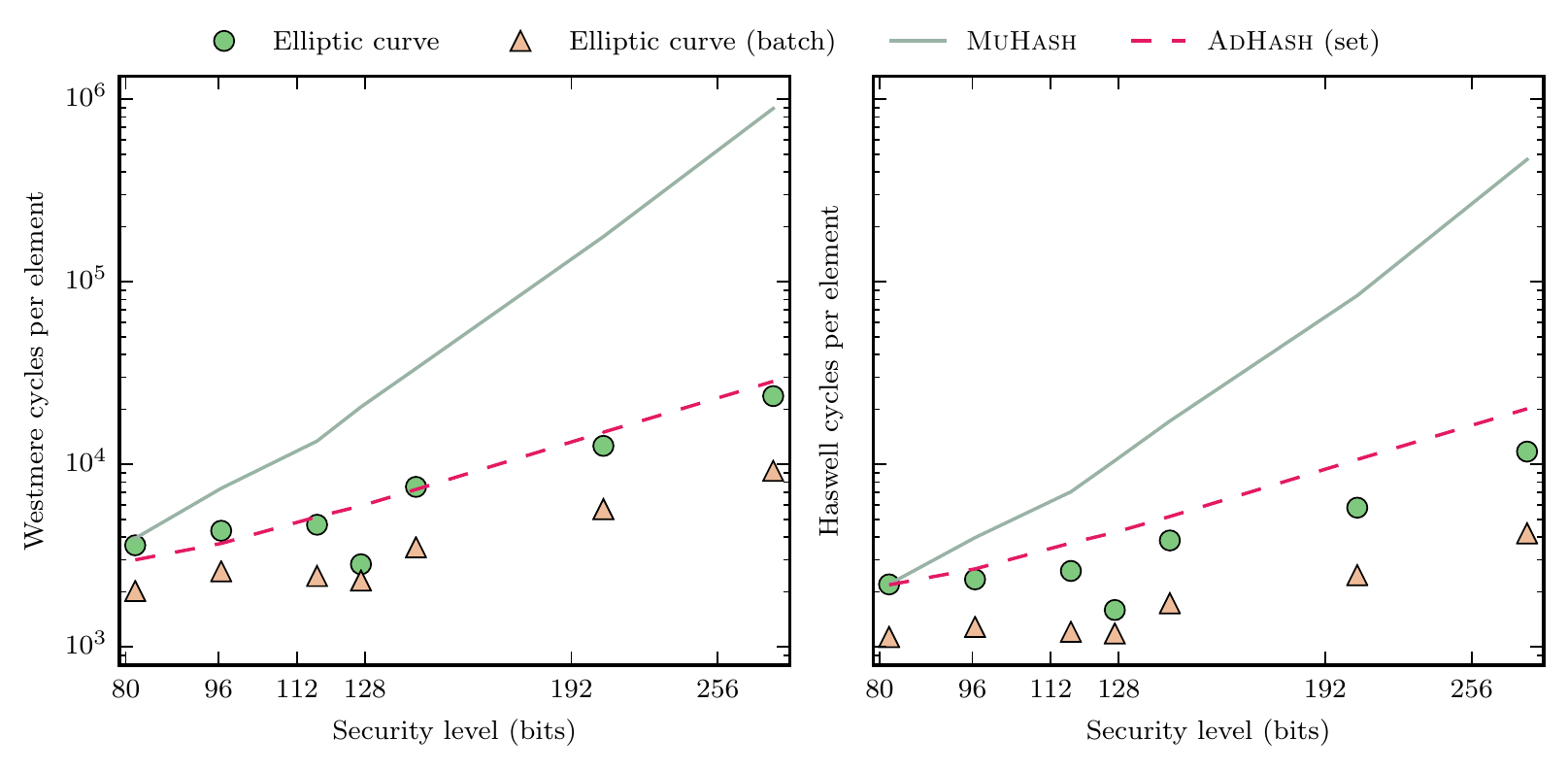}
  \caption{Comparison of multiset hashing performance at different security levels.  The elliptic curve corresponding to each security level is given in \cref{tab:multiset_hash_cycles}.  The security level for \adhash{} applies only to \emph{set} hashing, as described in \cref{sec:adhash-attack}.}
  \label{fig:multiset_hash_cycles}
\end{figure*}

For each multiset hash $H$, we measured the computational cost of incremental hash code updates corresponding to a sequence of incremental additions or removals of multiset elements, i.e.\ incrementing or decrementing by 1 the multiplicity of each element in the sequence.  Larger changes in multiplicity can also be handled efficiently by scalar multiplication in the group, but we expect incremental additions and removals to be the most common case.  We used a sequence of 1024 randomly generated 32-byte strings;\footnote{As ECMH depends on lookup tables with a block size of $B = 8$, 1024 random elements ensures high coverage of the tables and a random access pattern, in order to correctly estimate execution time, as described in \cref{sec:memory-dependent-measurement}.} longer strings would simply impose an additional cost independent of $H$.

The average cost per element reflects the cost of the intermediate hash function based on BLAKE2, the cost of encoding the expanded bit sequence as a group element, and the cost of one group operation to add the encoded element to a running total.  In the case of ECMH, the encoding is \textsc{SWChar2} and the group operation is implemented as the mixed addition of a $\lambda$-affine and a $\lambda$-projective point; batch ECMH effectively replaces 1 field inversion by 3 multiplications, as described in \cref{sec:batch-swchar2}.  In the case of \adhash{}, the encoding is trivial and the group operation is simply integer addition; batch computation would offer no advantage.  For \muhash{}, the encoding requires a comparison and at most one subtraction, and the group operation requires just a single Montgomery multiplication, as described in \cref{sec:muhash-montgomery}; batch computation would offer no advantage over the Montgomery representation already used.

The results are shown in \cref{fig:multiset_hash_cycles} and in \cref{tab:multiset_hash_cycles}.  Only element addition performance is shown, as due to the representations used, element removal performance is nearly identical.  Timings for point encoding, compression, and decompression are given in \cref{tab:elliptic_curve_operation2_cycles}.  Base field operation timings are given in \cref{tab:field_operation_cycles}, and a comparison of curve operation performance under $\lambda$-affine and $\lambda$-projective point representations is given in \cref{tab:elliptic_curve_operation_cycles}.

\begin{table*}[p]
  \caption{Comparison of multiset hashing performance, as in \cref{fig:multiset_hash_cycles}.  Note that the \adhash{} performance applies only to \emph{set} hashing.}
  \centerfloat
  \footnotesize

\begin{tabular}{llrc@{/}c@{/}c@{/}cccrc@{/}c@{/}c@{/}ccc}
& & & \multicolumn{6}{c}{Westmere cycles} & & \multicolumn{6}{c}{Haswell cycles}\\
\cmidrule(r){4-9}\cmidrule(r){11-16}& &  & \multicolumn{4}{c}{ECMH} &  &  &  & \multicolumn{4}{c}{ECMH} &  & \\

$n$&Curve &  & single & blind & batch & blind & \muhash{} & \adhash{} &  & single & blind & batch & blind & \muhash{} & \adhash{}\\
\midrule
81&sect163k1& & 3601&4436&2023&2418 & 3939 & 2998& & 2199&2556&1133&1349 & 2208 & 2186\\
96&sect193r1& & 4326&5444&2595&3198 & 7384 & 3687& & 2342&2755&1287&1580 & 3967 & 2674\\
116&sect233k1& & 4667&5726&2444&2933 & 13414 & 5160& & 2605&2968&1209&1495 & 7074 & 3708\\
119&sect239k1& & 5183&6361&2630&3164 & 16532 & 5117& & 3061&3474&1422&1700 & 8537 & 3684\\
127&GLS254& & 2835&3872&2307&2882 & 20631 & 5920& & 1592&1973&1184&1426 & 10472 & 4239\\
141&sect283k1& & 7524&9271&3513&4254 & 33472 & 7286& & 3828&4291&1733&2024 & 17251 & 5178\\
204&sect409k1& & 12621&16696&5686&6878 & 176767 & 14997& & 5788&6897&2473&2948 & 84027 & 10632\\
285&sect571k1& & 23654&29628&9206&10746 & 890172 & 28485& & 11745&16664&4188&4940 & 467938 & 20152\\

\end{tabular}
  \label{tab:multiset_hash_cycles}
\end{table*}

\begin{table*}[p]
  \caption{Performance of elliptic curve point encoding, compression (to a minimal-length bit string), and decompression (from said bit string).  Compression (comp.) and decompression (dec.) use $\lambda$-projective coordinates.  Batch encoding is with a batch size of $256$.}
  \centering
  \footnotesize

\begin{tabular}{lrc@{/}c@{/}c@{/}cccrc@{/}c@{/}c@{/}ccc}
 & & \multicolumn{6}{c}{Westmere cycles} & & \multicolumn{6}{c}{Haswell cycles}\\
\cmidrule(r){3-8}\cmidrule(r){10-15} &  & \multicolumn{4}{c}{\textsc{SWChar2}} &  &  &  & \multicolumn{4}{c}{\textsc{SWChar2}} &  & \\

Curve &  & single & blind & batch & blind & Comp. & Dec. &  & single & blind & batch & blind & Comp. & Dec.\\
\midrule
sect163k1& & 2629&3248&853&1234 & 2222 & 2268& & 1500&1854&432&641 & 1370 & 1402\\
sect193r1& & 3217&4073&1227&1821 & 2541 & 2669& & 1603&2014&548&823 & 1370 & 1443\\
sect233k1& & 3577&4340&1076&1537 & 3080 & 3154& & 1852&2214&510&719 & 1673 & 1732\\
sect239k1& & 4033&4880&1158&1670 & 3468 & 3569& & 2227&2616&570&817 & 2024 & 2045\\
GLS254& & 1671&2551&978&1566 & 1166 & 1245& & 874&1280&437&709 & 665 & 716\\
sect283k1& & 5738&6865&1587&2253 & 4772 & 5134& & 2822&3278&698&993 & 2320 & 2624\\
sect409k1& & 8612&10523&2707&3720 & 7497 & 7563& & 4395&5172&1157&1573 & 3883 & 4074\\
sect571k1& & 17968&21721&4301&5867 & 15174 & 16337& & 8987&13329&1867&2573 & 7132 & 8448\\

\end{tabular}
  \label{tab:elliptic_curve_operation2_cycles}
\end{table*}

\clearpage

\begin{table*}[p]
  \caption{Field operation performance for $\F_{2^m}$.  Batch inversion is with a batch size of $256$.}
  \centering
  \footnotesize
\begin{tabular}{lrccc@{/}c@{/}c@{/}cc@{/}crccc@{/}c@{/}c@{/}cc@{/}c}
 & & \multicolumn{8}{c}{Westmere cycles} & & \multicolumn{8}{c}{Haswell cycles}\\
\cmidrule(r){3-10}\cmidrule(r){12-19} &  &  &  & \multicolumn{4}{c}{Invert} & \multicolumn{2}{c}{QS} &  &  &  & \multicolumn{4}{c}{Invert} & \multicolumn{2}{c}{QS}\\

$m$ &  & Mul. & Sq. & single & blind & batch & blind & var. & blind &  & Mul. & Sq. & single & blind & batch & blind & var. & blind\\
\midrule
127& & 44 & 11 & 721&904&124&124 & 41&131& & 23 & 9 & 435&534&59&59 & 21&82\\
163& & 84 & 32 & 1807&2074&266&269 & 115&227& & 43 & 24 & 1159&1309&127&127 & 67&152\\
193& & 113 & 26 & 2210&2533&343&344 & 149&264& & 46 & 20 & 1129&1308&128&128 & 79&155\\
233& & 109 & 30 & 2745&3092&341&342 & 191&313& & 48 & 24 & 1439&1583&131&131 & 95&178\\
239& & 119 & 34 & 3139&3492&372&376 & 187&313& & 51 & 31 & 1755&1933&160&159 & 93&183\\
254& & 99 & 17 & 868&1211&310&313 & 88&245& & 38 & 15 & 514&664&111&116 & 62&158\\
283& & 148 & 36 & 4438&4869&473&474 & 420&560& & 55 & 28 & 2222&2423&175&179 & 227&296\\
409& & 274 & 35 & 7899&8688&884&867 & 807&959& & 93 & 30 & 3500&3805&291&296 & 404&478\\
571& & 431 & 65 & 16217&17808&1308&1323 & 1457&1692& & 168 & 38 & 6873&7611&464&489 & 737&842\\

\end{tabular}
  \label{tab:field_operation_cycles}
\end{table*}

\begin{table*}[p]
  \caption{Performance of elliptic curve group operations using $\lambda$-affine and $\lambda$-projective point representations.  The result of point addition or doubling is always represented in $\lambda$-projective coordinates.}
  \centering
  \footnotesize
\begin{tabular}{lrc@{/}c@{/}cc@{/}cc@{/}crc@{/}c@{/}cc@{/}cc@{/}c}
 & & \multicolumn{7}{c}{Westmere cycles} & & \multicolumn{7}{c}{Haswell cycles}\\
\cmidrule(r){3-9}\cmidrule(r){11-17} &  & \multicolumn{3}{c}{Add} & \multicolumn{2}{c}{Double} & \multicolumn{2}{c}{Negate} &  & \multicolumn{3}{c}{Add} & \multicolumn{2}{c}{Double} & \multicolumn{2}{c}{Negate}\\

Curve &  & aff. & mix. & full & aff. & proj. & aff. & proj. &  & aff. & mix. & full & aff. & proj. & aff. & proj.\\
\midrule
sect163k1& & 500&748&1016 & 192&472 & 12&18& & 213&305&408 & 105&188 & 6&9\\
sect193r1& & 604&952&1268 & 199&640 & 12&18& & 236&370&468 & 105&281 & 7&9\\
sect233k1& & 624&936&1276 & 202&540 & 12&18& & 235&341&462 & 107&224 & 7&10\\
sect239k1& & 684&1032&1388 & 244&620 & 12&18& & 308&453&595 & 144&311 & 6&9\\
GLS254& & 572&856&1168 & 162&488 & 9&14& & 219&308&435 & 78&196 & 5&9\\
sect283k1& & 864&1308&1792 & 280&768 & 15&23& & 329&489&655 & 138&302 & 12&19\\
sect409k1& & 1496&2320&3144 & 416&1280 & 18&29& & 542&788&1075 & 194&506 & 12&22\\
sect571k1& & 2276&3516&4772 & 644&1956 & 21&35& & 859&1253&1688 & 296&740 & 15&22\\

\end{tabular}
  \label{tab:elliptic_curve_operation_cycles}
\end{table*}

\clearpage

\section{Discussion}
\label{sec:discussion}

Elliptic curve multiset hash significantly outperforms the existing methods of
\muhash{} and \adhash{}, particularly in the batch setting, while requiring
significantly smaller hash codes at all security levels.  In fact, the hash code
size is essentially optimal.  Because the single field inversion required by the
encoding function \textsc{SWChar2} accounts for a large fraction of the
computational cost, particularly with larger field degrees, the use of
Montgomery's trick in the batch setting significantly reduces the computational
cost.  The lower computational cost at the 127-bit security level is due to the
efficiency of the GLS254 curve implementation; the quadratic extension field
representation of $\gf{254}$ employed, and the close match of the degree to the
word size $W = 128$, significantly reducing the cost of field operations.
Quadratic extension field representations for other fields, such as $\gf{502}$,
could potentially be used to obtain similar performance improvements at other
security levels. Furthermore, our choice of parameters follows the trend of
increasing native support to binary field arithmetic in Desktop
processors and will likely benefit from improvements to the carry-less
multiplication instruction in the recently released Broadwell processor family.

Our work is very related to the Encrypted Elliptic Curve Hash
(EECH)~\cite{brown2008encrypted}.  That construction also encodes
separate bit strings as points on a binary elliptic curve and then
combines those points using point addition.  Like our approach, it relies on the property of binary elliptic curves that curve points can be decoded from a non-redundant representation without expensive field exponentiations, using instead a precomputed lookup table for half-trace, and notes that better performance may be obtained using batch inversion and hardware support for carry-less multiplication.

The full EECH construction is proposed as an
incremental hash for bit strings (the message is split into fixed-size
blocks, and each block, concatenated with the block index, is encoded as
an elliptic curve point).  In contrast to our elliptic curve multiset hash, it is
specifically designed to \emph{avoid} reliance on an underlying random
oracle, relying instead on redundancy/padding in the point encoding
function for collision resistance.

While the \emph{full} construction is not well-suited to homomorphic multiset hashing\footnote{Using an elliptic curve over $\gf{m}$, under the EECH construction at most $b$ bits of input data can be encoded per point to retain collision resistance of $2^{m - b}$.  Optimal collision resistance of $2^m$ for the representation size requires that $b \leq m / 2$.  Each multiset element $a \in A$ (assumed to be a bit string) must therefore be split into one or more blocks of $b$ bits, each encoded as a separate elliptic curve point.  For elements longer than $b$ bits, this is likely to be significantly more expensive than hashing $a$ with a fast hash function like BLAKE2 and then encoding the result into a single elliptic curve point.  EECH also offers no preimage resistance by default.  There is a proposed pairing-based variant PEECH that relies on an elliptic curve pairing to define a homomorphic one-way function.  This provides preimage resistance at the cost of significantly higher computational cost and representation size.}, we can make the fairer comparison between our ECMH construction and a straightforward randomize-then-combine-style~\cite{bellare1997new} construction over binary elliptic curve groups using the implementation techniques proposed for EECH.  Such a construction was neither explicitly proposed nor implemented, and there was no prior evidence that it would be practical performance wise.  Our work goes significantly beyond this:
\begin{itemize}
\item We provide a thorough empirical analysis of performance, and demonstrate for the first time that an elliptic curve-based multiset hash actually significantly exceeds the performance of \adhash{} and \muhash{}.
\item We demonstrate that a fully blinded implementation is possible at only a minor performance penalty.  We also demonstrate batch variants of both the regular and fully-blinded implementations that are significantly faster.  In contrast, the try-and-increment encoding method proposed for EECH has no guaranteed time bound, making it unavoidably susceptible to timing attacks, and less amenable to speedup by batch inversion.
\item Our security proof is based on existing techniques~\cite{bellare1997new,brier2010efficient} but the security bound we obtain is novel in several ways:
  \begin{itemize}
  \item The hash function $\hat{H}$ into the elliptic curve group need not be indistinguishable from a random oracle, but is instead permitted to satisfy the weaker property of being an $(\alpha, \beta)$-weak encoding, which significantly reduces the computational cost.
  \item The hash function $\hat{H}$ can map to the full elliptic curve group, rather than only a cyclic subgroup, as is required by EECH.  This allows for a simpler implementation that does not rely on patent-encumbered techniques~\cite{brown2007method} for efficiently testing for subgroup membership.
  \end{itemize}
\end{itemize}

It was originally suggested~\cite{bellare1997new} that while finding collisions in \muhash{} is provably as hard as the Discrete Logarithm Problem (DLP), the converse is not necessarily true: it may be that \muhash{} is still collision resistant even if discrete logarithms can be computed efficiently.  In fact, though, by computing discrete logarithms, finding a collision in \muhash{} can be reduced to finding a collision in \adhash{}.  It would therefore be susceptible to a generalized birthday attack~\cite{wagner2002generalized} in the set hashing setting or to lattice reduction attacks in the multiset hashing setting.  The same reduction applies to our elliptic curve multiset hash, and is even more effective because of the smaller group order.


\appendix

\section{Security reduction based on ($\alpha, \beta$)-weak encodings}
\label{sec:proofthmdlpreduction}

We prove \cref{thm:dlpreduction}, which reduces solving discrete logarithms to finding collisions in a homomorphic multiset hash function based on an $(\alpha, \beta)$-weak encoding.

\dlpreduction*
\begin{proof}
  Let a $Q \in \langle P \rangle$, for which we wish to find $n \in \Z$ such that $n \cdot P = Q$, be given.  We simulate each successive distinct query $h(a_i)$ to the random oracle $h$ for $i = 1, \ldots, k$ using the following algorithm:
  \begin{enumerate}
  \item Sample uniformly at random $r_i \in \Z_\rho$, $d_i \in \Set{0,1}$, $J_i \in \overline{\langle P \rangle}$, $j \in \Z_{\lceil \alpha \rceil}$.
  \item Compute $Q_i = r_i Q + d_i P + J_i$.  Note that since $\langle P \rangle$ has prime order, $Q$ is a generator of $\langle P \rangle$, and therefore $Q_i$ is distributed uniformly in $G$.
  \item If $j < \card{f^{-1}(Q_i)}$, sample $x_i$ from $f^{-1}(Q_i)$ uniformly at random.  Otherwise, resample $r_i, d_i, J_i$, and $j$.
  \item Return $x_i$.  Note that $x_i$ is uniformly distributed in $X$, and the expected number of sampling attempts is $\alpha / \beta$.
  \end{enumerate}

  Under the simulated $h$, $\mathcal{C}$ finds a non-empty $M \in \ker H$ in expected time $t$ with success probability $\epsilon$.  Consider the case that a collision is found.  (Otherwise, we fail to compute the discrete logarithm.)  Without loss of generality, we can assume $M$ is non-zero only for values $a_i$ on which $h$ was queried.  Thus, we have
  \begin{align*}
    0_G &= \sum_{i=1}^k M(a_i) \cdot h(a_i)
        = \sum_{i=1}^k M(a_i) \cdot Q_i \\
        &= \sum_{i=1}^k M(a_i) \cdot \left[ r_i \cdot Q + d_i \cdot P + J_i \right],
  \end{align*}
  which implies
  \begin{equation}
    \label{eq:dlpreduction:eq-full-with-subs}
    r \cdot Q + \sum_{i=1}^k J_i \cdot M(a_i) = - d \cdot P,
  \end{equation}
  where
  \begin{align*}
    r &= \sum_{i=1}^k r_i \cdot M(a_i) \bmod \rho, &&
    d = \sum_{i=1}^k d_i \cdot M(a_i) \bmod \rho.
  \end{align*}
  Since $r \cdot Q \in \langle P \rangle$ and $-d \cdot P \in \langle P \rangle$, it follows that $\sum_{i=1}^k J_i \cdot M(a_i) = O_G$ in \cref{eq:dlpreduction:eq-full-with-subs}; we therefore have $r \cdot Q = - d \cdot P$.  Since $M$ is non-empty, there exists a value $i$ such that $M(a_i) \not= 0$.  Consider that the distribution of $d_i$ conditioned on $Q_1, \ldots, Q_k$ is still uniform in $\Set{0,1}$, and therefore $\Pr(d_i = 0) = 1/2$, and hence, $\Pr(d = 0) \leq 1/2$.  If $d \not= 0$, then $r \not= 0$, and therefore $\Pr(r \not= 0) \geq 1/2$.

  If $r = 0$, we fail to compute the discrete logarithm.  Otherwise, $r$ has an inverse $r^{-1}$ in $\Z_\rho^\times$ and we have $Q = r^{-1} r Q = - r^{-1} d P$.  Thus, $n = - r^{-1} d$ is a solution to the discrete logarithm problem.  Since we only fail if $\mathcal{C}$ fails or $r = 0$, we find a solution with probability at least $\epsilon/2$.

Each query $a_i$ to the simulated random oracle requires a table lookup to check if $a_i$ has been queried previously.  If it has not, we must repeatedly sample $r_i$, $d_i$, $J_i$ and $j$ and compute $Q_i = r_i Q + d_i P + J_i$ in time
\ifcompj
\begin{multline*}
T_3 = T_{\mathrm{samp}}(\Z_\rho) + T_\mathrm{samp}(\Z_2) +\\ T_\mathrm{samp}(\overline{\langle g \rangle}) + T_{\exp(G)} + 2 T_\mathrm{mult}(G),
\end{multline*}
\else
\[
T_3 = T_{\mathrm{samp}}(\Z_\rho) + T_\mathrm{samp}(\Z_2) + T_\mathrm{samp}(\overline{\langle g \rangle}) + T_{\exp(G)} + 2 T_\mathrm{mult}(G),
\]
until $j < \lvert f^{-1}(Q_i) \rvert$, which requires $\beta$ attempts in expectation, since $f$ is an $(\alpha, \beta)$-weak encoding.  We then sample $x_i \in f^{-1}(Q_i)$.  Thus, each of the $q$ queries to the random oracle require expected time $\beta T_3 + T_2$, where
\[ T_2 = T_\mathrm{lookup} + T_\mathrm{samp}(f^{-1}). \]

We can compute $r$ and $d$ as a sum of $L$ terms in time $L \cdot T_4$, where
\[ T_4 = T_\mathrm{lookup} + 2 T_\mathrm{add}(\Z_\rho) + T_\mathrm{mult}(\Z_\rho). \]
Finally, we can compute $n$ from $r$ and $d$ in time
\[ T_1 = T_\mathrm{inv}(\Z_\rho) + T_\mathrm{mult}(\Z_\rho) + T_\mathrm{negate}(Z_\rho). \]
Thus, the total expected time is $t + T_1 + q T_2 + q \beta T_3 + L T_4$.
\end{proof}


\section{Security analysis of $\mbox{\adhash{}}$ in the multiset setting}
\label{sec:adhash-attack}

The best known attack on Bellare and Micciancio's incremental hash
function \adhash{} when it is used to hash \emph{sets} is Wagner's
generalized birthday attack~\cite{wagner2002generalized}. However,
when the function is used for \emph{multiset} hashing, as proposed by
Clarke~{et al.}~\cite[Theorem 6]{clarke2003incremental}, its security is
much weaker. Indeed, finding a multiset collision on \adhash{} with
$q$ random oracle queries is equivalent to finding a vector
$(a_1,\dots,a_q)\in\Z^q$ of polynomial norm such that:
\[ \sum_{i=1}^q a_i h_i \equiv 0 \pmod M, \]
where the $h_i$'s are the hash values returned by the oracle, and $M$ is
the \adhash{} modulus. In other words, the problem is to find a short
vector in the full rank lattice $L\subset\Z^q$ of vectors orthogonal to
$(h_1,\dots,h_q)$ modulo $M$.

The volume $\vol(L) = [\Z^q:L]$ of $L$ is clearly at most $M$, since $L$
is the kernel of a homomorphism to $\Z/M\Z$. Therefore, a lattice
reduction algorithm with Hermite factor constant $c$
(see~\cite{DBLP:conf/eurocrypt/GamaN08}) is expected to find a vector in
$L$ of Euclidean norm at most $c^q\cdot M^{1/q}$. By choosing $q =
\sqrt{\frac{\log M}{\log c}}$, we obtain a multiset collision of size
roughly $2\sqrt{\log_2 M\cdot\log_2 c}$ bits. For $k$ bits of security
against this multiset collision attack, it is thus necessary to choose:
\[ \log_2 M \geq \frac{k^2}{4\log_2 c}. \]
This is similar to Wagner's attack in the sense that the size of $M$
should be at least quadratic in the security parameter, but the constant
is typically much larger. Over a large range of lattice dimensions,
a security level of $k=128$ bits corresponds to a Hermite factor constant
$c\approx1.007$~\cite{chen2011bkz,DBLP:conf/ima/PolS13}.
Hence, a conservative choice of $M$ should be at least 400,000-bit long,
which is obviously impractical. Even $k=80$ corresponds to
$c\approx1.008$ and requires $M$ to be chosen larger than 100,000 bits.

At any rate, recommended sizes for the set-hash setting are highly
insecure in the multiset hash setting. Consider a modulus $M$ of $1600$
bits, appropriate for $80$-bit security in the set-hash setting. Simply
doing $q=230$ oracle queries and easily reducing the corresponding
lattice with LLL (not even BKZ!), which has a Hermite factor constant
$c\approx1.021$, yields a multiset collision of weight about
$c^{230}\cdot 2^{1600/230}\leq 15000$ (less that $14$-bit long).
Similarly, given a $4096$-bit modulus $M$ (as used for $128$-bit security
in the set-hash setting), doing $q=500$ queries and reduction the
corresponding $500$-dimensional lattice with BKZ-28\footnote{This is by
no means a large computational effort even by academic standards. Recent
academic lattice reduction records target lattices of dimension $>800$
using BKZ with block size $90$ and
up~\cite{chen2011bkz,lindner2014latticechallenge}.}, which has a Hermite
factor constant $c\approx1.011$~\cite{DBLP:conf/eurocrypt/GamaN08},
yields a multiset collision of weight about $c^{500}\cdot 2^{4096/500}\leq
70000$ (less than $17$-bit long).

\section{Group structure implied by incremental additions}
\label{sec:monoid-to-group}

Consider a more limited definition of an incremental multiset hash function, under which only incremental additions (and non-negative multiplicities) are supported:
\begin{definition}
  Let $A$ be a set, and let $T$ be a finite set with an associative
operation $+_T \colon T \times T \rightarrow T$.  A function $H \colon
\Z_{\geq 0}^A \rightarrow T$ is a \emph{monoid-homomorphic multiset hash
function} if $H(M_1 + M_2) = H(M_1) +_T H(M_2)$ for all $M_1, M_2 \in
\Z^{(A)}$.
\end{definition}

Note that $(H(\Z_{\geq 0}^A), +_T)$ is necessarily a commutative monoid under this definition.  Thus, without loss of generality, we can assume that $(T, +_T)$ is a commutative monoid.

\begin{theorem}
  If we make the additional assumption that $(T, +_T)$ has the
cancellation property, i.e.\ $a + b = a + c$ implies $b = c$ for all $a,
b, c \in T$, then we can construct a (group-)homomorphic multiset hash
function $H'$ from $\Z^{(A)}$ into a group $G$ that embeds $T$.  Furthermore, this construction has only a constant factor time and space overhead of $2$.
\end{theorem}
\begin{proof}
  Since $T$ is a finite, commutative monoid with the cancellation property, there must exist an inverse for every element, and therefore $T$ is a group.
  However, to ensure that the inverse can be computed efficiently, we use the Grothendieck construction in which we represent the positive and negative parts by separate elements of $T$.

  Let $G$ be the quotient set $T \times T /\!\!\!\equiv_G$, where the equivalence relation $\equiv_G$ is given by $(a_+, a_-) \equiv_G (b_+, b_-)$ if, and only if, $a_+ + b_- = a_- + b_+$, for all $a_+, a_-, b_+, b_- \in T$.  We define the addition operation $[(a_+, a_-)] +_G [(b_+, b_-)] = [(a_+ +_T b_+, a_- +_T b_-)]$.  Note that $+_G$ respects $\equiv_G$, and the inverse is given by $-[(a_+, a_-)] = [(a_-, a_+)]$.

  We define the hash function $H' \colon \Z^{(A)} \rightarrow G$ by $H'(M) = [(H(\max(M,0)), H(\max(-M,0)))]$.
  Since $\max(-M_1,0) + \max(-M_2,0) + \max(M_1 + M_2,0) = \max(M_1,0) +
\max(M_2,0) + \max(-(M_1 + M_2),0)$ for all $M_1, M_2 \in \Z^{(A)}$, we have
$H'(M_1 + M_2) = [(H(\max(M_1 + M_2,0)), H(\max(-(M_1 + M_2))))] 
               = [(H(\max(M_1,0)), H(\max(-M_1,0)))] + [(H(\max(M_2,0)), H(\max(-M_2,0)))]
               = H'(M_1) + H'(M_2)$.

  Finally, we can embed $T$ in $G$ using that map $\phi(a) = [(a,H(\emptyset))]$ for all $a \in T$.  It follows directly from the definition of $\equiv_G$ and $+_G$ that $\phi$ is an injective homomorphism.  Note that the representation size for an element of $G$ is twice the representation size of an element of $T$, and $H'$ and $+_G$ require two invocations of $H$ and $+_T$, respectively.
\end{proof}

\section{Equivalence of incremental multiset hash function definitions}
\label{sec:clarke-equivalence}

\Cref{def:homomorphic-multiset-hash} is based on the definition of an
incremental multiset hash function given by Clarke~{et al.}~\cite{clarke2003incremental}, which we restate as follows:

\begin{definition}
  \label{def:clarke-multiset-hash}
  Let $\mathcal{H}^r : A^{\Z_{\geq 0}} \rightarrow T$ and $+_\mathcal{H}^r \colon T \times T \rightarrow T$ be probabilistic algorithms using randomness $r \in R$, where $T$ is a finite set, and let $\equiv_\mathcal{H}$ be an equivalence relation over $T$.  The triple $(\mathcal{H}, +_\mathcal{H}, \equiv_\mathcal{H})$ is a \emph{multiset hash function} if it satisfies the following properties:
  \begin{enumerate}
  \item $\mathcal{H}^{r_1}(M) \equiv_\mathcal{H} \mathcal{H}^{r_2}(M)$,
for all $M \in \Z^{(A)}$, $r_1, r_2 \in R$;
  \item $+_\mathcal{H}$ respects the equivalence relation $\equiv_{\mathcal{H}}$;
  \item $s_1 +_\mathcal{H}^{r_2} s_2 \equiv_\mathcal{H} s_3$ if
$\mathcal{H}^{r_3}(M_1) \equiv_\mathcal{H} s_1$, $\mathcal{H}^{r_4}(M_2)
\equiv_\mathcal{H} s_2$, and $\mathcal{H}^{r_1}(M_1 + M_2) = s_3$ for all
$M_1, M_2 \in \Z^{(A)}$, $s_1, s_2, s_3 \in T$, $r_1, r_2, r_3, r_4 \in R$.
  \end{enumerate}
\end{definition}

This differs from our definition of a monoid-homomorphic multiset hash function (\cref{sec:monoid-to-group}) only in that it allows for randomness in the hash function and in the addition operation $+_T$.  Note that this randomness is for a fixed hash function, and is independent of the randomness in choosing the hash function from a hash function family.  The multiset hash function \msetaddhash{}~\cite{clarke2003incremental} relies on this randomness for security.  In fact, though, the randomness is not integral to the hashing operation itself, but rather is used as a nonce in encrypting the hash code, which we view as an orthogonal operation.\footnote{Note also that \msetaddhash{} is secure as a keyed hash function but not (under the same assumptions) as a public hash function.}  Therefore, we dispense with this randomness in our definition.

As explained in \cref{sec:monoid-to-group}, if we assume that $(T, +_T)$ has the cancellation property, i.e.\ that $+_T$ does not itself introduce any additional collisions, then a simple construction produces a (group-)homomorphic multiset hash function from any multiset hash function satisfying \cref{def:clarke-multiset-hash}.

\end{document}
